%% file: paper.tex
\documentclass[letterpaper,twocolumn,10pt]{article}
\usepackage{usenix}

\usepackage{hyperref}

\usepackage{graphicx}
\usepackage{booktabs}
\usepackage{etoolbox}
\usepackage{appendix}
\usepackage[linesnumbered,ruled,vlined]{algorithm2e}

\usepackage{enumitem,amssymb}

\usepackage{glossaries}
\glsdisablehyper

\usepackage{subcaption}
\usepackage{xspace}

\usepackage{amsthm}
\usepackage{amsmath}
\usepackage[utf8]{inputenc}

\usepackage[capitalise]{cleveref}
\makeatletter
\newcommand{\crefnames}[3]{%
  \@for\next:=#1\do{%
    \expandafter\crefname\expandafter{\next}{#2}{#3}%
  }%
}
\makeatother
\crefnames{part,chapter,section}{\S}{\S\S}
\crefformat{section}{\S#2#1#3}
\Crefformat{section}{\S#2#1#3}
\crefmultiformat{section}{\S\S#2#1#3}{ and~#2#1#3}{, #2#1#3}{, and~#2#1#3}
\Crefmultiformat{section}{\S\S#2#1#3}{ and~#2#1#3}{, #2#1#3}{, and~#2#1#3}

\usepackage{xcolor}

\patchcmd{\abstract}{\vspace*{...}}{}{}{}
\usepackage{titlesec}

\input{sections/mycommands}

\begin{document}

\title{
\Large \bf Atlas: Efficient Verifiable Semantic Search \\
}
\author{
{\rm Nikolay Avramov\textsuperscript{1}, Hidde Lycklama\textsuperscript{1}, Alexander Viand\textsuperscript{2}, Anwar Hithnawi\textsuperscript{1}}  
\\
{\textsuperscript{1}\textit{University of Toronto} \ \textsuperscript{2}\textit{Belfort Labs} }
\vspace{12pt}
}

\makeatletter
\patchcmd{\maketitle}
	{\@maketitle}
	{\vspace{-2em}\@maketitle\vspace{-1em}}%
	{}
	{}
\makeatother

\date{}

\maketitle

\begin{abstract}
Semantic search is a core primitive of modern applications, powering recommender systems, web search, and retrieval-augmented generation for language models. The provider controls the index and query execution, leaving clients to trust that results come from the right algorithm over the intended index. A provider may truncate search to cut cost, bias results, or otherwise deviate from the specified execution undetected. Verifiability can remove this trust assumption by proving that results follow the agreed algorithm over a committed index. Realizing this efficiently is hard, as retrieval at scale relies on HNSW, a graph-based algorithm whose data-dependent traversal maps poorly onto the fixed constraint systems of zero-knowledge proofs. Prior verifiable systems therefore target regular, cluster-based indices that are easier to encode, sacrificing the recall of graph-based search. We present \oursystem, a system that lets a provider prove a query was answered correctly against its committed index without revealing the index. At its core is a new zero-knowledge proof for HNSW search, built on three techniques: preprocessing that shifts all database-dependent cost offline, so per-query proving scales with the traversal rather than the database; a restructuring of HNSW into a fixed-size-state procedure that we prove returns the same result; and a timestep-tagged batching that merges the per-step arguments of the entire traversal into one. \oursystem is the first to demonstrate verifiable graph-based search at scale, proving a query in under a second on the SIFT1M benchmark and in 2.0 seconds at 100 million vectors, while maintaining the recall of plaintext HNSW and revealing nothing about the index beyond the result. In a complete RAG pipeline, \oursystem' proven retrieval preserves end-to-end answer quality, and reaches higher quality at lower proving cost than all prior verifiable retrieval systems.
\end{abstract}

\glsresetall

\input{sections/introduction}

\input{sections/background}
\input{sections/relatedwork}
\input{sections/overview}

\input{sections/design}
\input{sections/evaluation}
\begin{appendices}

  \crefalias{section}{appendix}%
  \crefname{appendix}{Appendix}{Appendices}%
  \Crefname{appendix}{Appendix}{Appendices}%

\bibliographystyle{plain}
\interlinepenalty=10000 %
\bibliography{references}

\input{sections/appendix}

\end{appendices}

\end{document}

%% file: sections/mycommands.tex
\usepackage{xspace}

\newcommand{\oursystem}{Atlas\xspace}

\newcommand{\fakeparagraph}[1]{\vskip 5pt\noindent\textbf{#1. }}

\usepackage{xurl}
\usepackage{aliascnt}
\usepackage{setspace}
\SetAlCapFnt{\small}
\SetAlgoNlRelativeSize{-1}
\SetAlgoInsideSkip{smallskip}
\SetAlCapSkip{2pt}
\SetAlgoSkip{}
\SetAlgoInsideSkip{}
\newaliascnt{theorem}{definition}
\newtheorem{theorem}[theorem]{Theorem}
\aliascntresetthe{theorem}
\crefname{theorem}{theorem}{theorems}
\Crefname{theorem}{Theorem}{Theorems}

\newaliascnt{lemma}{definition}
\newtheorem{lemma}[lemma]{Lemma}
\aliascntresetthe{lemma}
\crefname{lemma}{lemma}{lemmas}
\Crefname{lemma}{Lemma}{Lemmas}

\newcommand{\polyid}{\ensuremath{\mathit{id}}\xspace}
\newcommand{\polyl}{\ensuremath{\ell}\xspace}
\newcommand{\polyd}{\ensuremath{d}\xspace}
\newcommand{\di}[2]{\ensuremath{\mathsf{d}(#1, #2)}\xspace}

\newcommand{\polye}{\ensuremath{e}\xspace}
\newcommand{\polysel}{\ensuremath{\textit{sel}}\xspace}

\newcommand{\polycan}{\ensuremath{\mathit{id}_\text{can}}\xspace}
\newcommand{\polynb}{\ensuremath{\mathit{id}_\text{nb}}\xspace}

\newcommand{\polyidrest}{\ensuremath{\mathit{id}_\text{rest}}\xspace}
\newcommand{\polydrest}{\ensuremath{\mathit{d}_\text{rest}}\xspace}

\newacronym{zkp}{ZKP}{zero-knowledge proof}
\newacronym[longplural=zero-knowledge Succinct Non-Interactive Arguments of Knowledge]{zksnark}{zk-SNARK}{zero-knowledge Succinct Non-Interactive Argument of Knowledge}
\newacronym{piop}{PIOP}{Polynomial Interactive Oracle Proof}
\newacronym{pcs}{PCS}{Polynomial Commitment Scheme}

\newcommand{\Eq}{\mathsf{Eq}}
\newcommand{\Mem}{\mathsf{Mem}}
\newcommand{\restr}[2]{#1\big|_{#2}}
\newcommand{\Selg}{\mathsf{Sel}_g}
\newcommand{\Nbg}{\mathsf{Nb}_g}
\newcommand{\Selb}{\mathsf{Sel}_b}
\newcommand{\Canb}{\mathsf{Can}_b}
\newcommand{\Nbb}{\mathsf{Nb}_b}
\newcommand{\Canbm}{\mathsf{Can}_b^{-}}
\newcommand{\selq}[1]{q_{#1}}
\newcommand{\polytag}{\tau}

\newcommand{\polydsel}{\mathit{d}_{\mathrm{sel}}}

\newcommand{\polyidrep}[1]{\mathit{id}^{(#1)}}
\newcommand{\polylrep}[1]{\ell^{(#1)}}
\newcommand{\polydcan}{\mathit{d}_{\mathrm{can}}}
\newcommand{\polydnb}{\mathit{d}_{\mathrm{nb}}}
\newcommand{\polydcansub}{\tilde{d}_{\mathrm{can}}}
\newcommand{\polydnbsub}{\tilde{d}_{\mathrm{nb}}} 
\newcommand{\polydworst}{\mathit{d}_{\mathrm{worst}}}

\newcommand{\Temb}{T_{\mathrm{emb}}}
\newcommand{\Ted}{T_{\mathrm{ed}}}
\newcommand{\Trange}{T_{\mathrm{range}}}
\newcommand{\dinf}{d_{\infty}}

\newcommand{\polyproc}{\mathit{proc}}
\newcommand{\polycons}{\mathit{cons}}
\newcommand{\polyprocrest}{\mathit{proc}_{\mathrm{rest}}}
\newcommand{\polyidsel}{\mathit{id}_{\mathrm{sel}}}

\crefname{table}{table}{tables}
\Crefname{table}{Table}{Tables}

%% file: sections/introduction.tex
\section{Introduction}

Recommender systems, web search, and retrieval-augmented generation (RAG) for large language models increasingly rely on semantic search~\cite{lewis2020rag}. 
Semantic search embeds items and queries as vectors and returns the nearest items to the query in this space. Because proximity is intended to capture semantic similarity, it can retrieve relevant content even when it has little lexical overlap with the query.
Approximate nearest-neighbor (ANN) search~\cite{hnsw, aumuller2020annbenchmarks} makes this practical at scale: instead of computing exact nearest neighbors, it returns sufficiently close points, enabling low-latency search over billions of vectors. This ability to retrieve by meaning at massive scale has made semantic search a default mechanism for finding relevant content in modern systems.

Retrieval-grounded AI systems are only as reliable as the passages they retrieve: production systems have returned incorrect answers after surfacing stale or erroneous content~\cite{microsoft-bug-report}, and attackers have exploited the retrieval path, planting documents that hijack generation once retrieved~\cite{some-prompt-injection-rag-attack}. Semantic search also plays a role in content provenance, where it is used to detect copyrighted material in AI pipelines or attribute generated content to rightholders~\cite{sureel.ai,sureel-patent-12567396,sureel-patent-12554767,cite-prorata.ai,copyleaks}. In such settings, the incentives to deviate from the intended protocol can be substantial. ANN search likewise underpins recommendation systems for social networks, e-commerce, advertising, and dating platforms~\cite{cite-snowflake,cite-databrick-hnsw}.
In each of these settings, the search is executed by a third-party provider that controls both the index and the computation, while the client sees only the returned results. This opacity leaves room for failures and/or deliberate deviations. A provider may serve stale results, truncate the search to cut costs, bias the ranking toward preferred outcomes, or quietly depart from the specified algorithm. These risks compound as semantic search is increasingly embedded in automated pipelines spanning multiple services that run on third-party infrastructure, where retrieved results are consumed directly by downstream services with no human oversight. These concerns have already drawn public scrutiny and regulatory interest~\cite{eu-dsa,senate-bill-5339}. Clients therefore need a way to verify semantic-search correctness: a guarantee that the returned set is exactly the one produced by the specified algorithm over an index to which the provider committed in advance.

\Glspl{zkp}~\cite{goldwasser1985knowledge, groth16} can, in principle, provide such a guarantee.
The provider generates a succinct proof that the returned search result is correct, which the client can verify without learning anything about the database beyond the result itself.
Realizing this for semantic search, however, is challenging. The de facto standard for ANN search at scale is HNSW~\cite{hnsw}. HNSW is control-flow heavy by construction: it navigates a layered proximity graph using priority queues, data-dependent branching, and early exits. Such computation maps poorly onto the arithmetic constraint systems underlying modern ZKPs, which require circuits to be fixed in advance and provisioned for worst-case behavior at every step.

To overcome this, prior verifiable search systems avoid HNSW and instead build on cluster-based indices~\cite{v3db,verirag}.
These methods partition the dataset into clusters offline and answer a query by searching the clusters nearest to it.
This fixed, data-independent search pattern encodes compactly as polynomial constraints.
This data-independence, however, sacrifices recall, as the clusters are selected once and a closer neighbor outside them is unreachable for the remainder of the search.
HNSW, by contrast, selects nodes throughout the search, conditioning each step on measured distances to the query, and its candidates therefore keep improving whereas a cluster-based search is confined to its initial selection.
Cluster-based systems thus trade retrieval quality for provability, and the adaptivity that quality requires is exactly what a fixed constraint system struggles to capture.

In this paper, we present \oursystem, an efficient zero-knowledge proof system for HNSW search.
\oursystem enables a provider to prove that a query was answered correctly against a committed index, while retaining the graph-based traversal on which modern retrieval systems rely.
The proven search matches the recall of plaintext HNSW at lower proving cost than all existing verifiable search systems, cluster-based and graph-based alike.
\oursystem achieves these results through three key ideas.
\emph{(i) Preserving HNSW's sublinear per-query cost.}
HNSW answers a query by visiting only $O(\log N)$ of the index's $N$ nodes.
Every visited node and traversed edge, however, must be proven consistent with the committed index.
A naive consistency check reads the entire committed index, and the $O(\log N)$ search then incurs a linear proving cost.
We decouple the consistency check from the traversal and defer its database-dependent cost to a preprocessing phase via cq lookup arguments~\cite{cq}.
The per-query proving cost then scales with the length of the traversal rather than the size of the database, preserving HNSW's sublinearity.
\emph{(ii) Restructuring HNSW.}
We eliminate nearly all data-dependent control flow from HNSW by reformulating search as a procedure over fixed-size state and a fixed number of steps.
Concretely, we collapse the multi-layer descent into a single graph walk and replace the two interdependent priority queues with one bounded candidate set.
This reformulation substantially shrinks the trace to be proven, as each step merges an entire neighborhood into the set in a single batched update rather than one node at a time.
These structural changes preserve the search exactly, as we prove the reformulated output equivalent to HNSW's.
The step bound is the only point at which the reformulated search differs from HNSW, with sufficiently large budgets recovering HNSW's result exactly and smaller ones lowering proving cost at the expense of recall.
\emph{(iii) Compacting the per-step constraints.} 
A naive encoding instantiates separate permutation and lookup arguments for each traversal step over the corresponding region of the witness. Each instantiation introduces its own auxiliary polynomials and polynomial identities, causing the proving overhead to grow with the step budget.
We introduce timestep-tagged arguments, which tag every tuple with a step identifier and merge the per-step arguments of the entire traversal into a single invocation.

\paragraph{Evaluation Summary.} 
We implement \oursystem  and evaluate it on six standard ANN benchmarks, spanning one million to one hundred million vectors and embedding dimensions from 96 to 960, as well as on end-to-end RAG pipelines with learned text embeddings (\S\ref{sec:evaluation}). 
\oursystem proves a query over SIFT1M in 0.8 seconds, with a 16.5 kB proof that the client verifies in 40 milliseconds.
Proving cost scales with the beam width and step budget needed to reach a recall target, not with database size, and the 100$\times$ larger BIGANN-100M raises proving time only 2.5$\times$, to 2.0 seconds.
Truncation to a fixed step budget sacrifices little recall, as at 95th-percentile budgets the proven search stays within 0.8 recall@1 points of plaintext HNSW on integer datasets and every dataset exceeds 0.9 recall@1 at a deployable configuration.
In a complete RAG pipeline, \oursystem preserves this quality end to end, attaining higher answer quality at lower proving cost than prior verifiable retrieval systems.

%% file: sections/background.tex
\section{Background}

This section reviews the building blocks of our construction: zk-SNARKs from polynomial interactive oracle proofs (\S\ref{sec:bg-snarks}), the permutation and lookup arguments (\S\ref{sec:bg-arguments}), and the interface extensions used throughout (\S\ref{sec:bg-extensions}).

\subsection{zkSNARKs from Polynomial IOPs}
\label{sec:bg-snarks}%
A \gls{zkp} is a protocol in which a prover convinces a verifier that a statement holds without revealing anything beyond its validity.
A \gls{zksnark} is a \gls{zkp} that is additionally succinct and non-interactive, producing a single short proof that a computation on private inputs was executed correctly.
Concretely, the prover proves knowledge of a witness $w$ for a public statement $x$ of an NP relation $R$, i.e.\ $(x, w) \in R$, without revealing $w$ to the verifier.
A standard construction of \glspl{zksnark} combines a \gls{piop} with a \gls{pcs}~\cite{bfs20,plonk}.
In a \gls{piop}, the prover's messages are polynomials, provided to the verifier as oracles that can be queried at a small number of points.
A PCS instantiates the oracles with commitments: the prover sends a short commitment in place of each polynomial and answers each query with an evaluation proof.

A \gls{piop} enforces the correctness of a computation as a set of polynomial identities representing its execution trace.
We focus on univariate \glspl{piop}, in which the trace is encoded as evaluations of polynomials over a multiplicative subgroup $H = \{1, \omega, \omega^2, \dots, \omega^{n-1}\} \subseteq \mathbb{F}$ of order $n$.
Each length-$n$ sequence of trace values $a_0, \dots, a_{n-1}$ is interpolated as the unique polynomial $p$ of degree below $n$ with $p(\omega^i) = a_i$.
The identities take the form $g(p_1(X), \dots, p_k(X)) = 0$, required to hold at every point of $H$.
An identity holds over all of $H$ exactly when its polynomial is divisible by the vanishing polynomial $Z_H(X) = X^n - 1$, and the committed quotient serves as the proof, verified at a random evaluation point.

Every computation defines its own trace polynomials, but the constraints imposed on them are largely generic and the literature formalizes the recurring ones as self-contained arguments with their own auxiliary polynomials and identities.
Many recurring constraints reduce to two relations between multisets of values, enforced by the \emph{permutation argument}, proving two sequences are reorderings of each other, and the \emph{lookup argument}, proving every value of one sequence occurs in a second~\cite{plonk,haboeck2022logup}. 
The \emph{copy constraint}, a specialized instance of the permutation argument, is the central mechanism connecting the polynomial identities of modern PIOPs by enforcing the consistency of a value's copies across them~\cite{plonk}.
The lookup argument in turn replaces constraints that are expensive to encode arithmetically with membership in a precomputed table, as in range checks~\cite{plookup,haboeck2022logup}.

\subsection{Permutation and Lookup Arguments}
\label{sec:bg-arguments}
Permutation and lookup arguments reduce claims about multisets to polynomial identities by encoding each multiset as a product over its elements.
Concretely, a multiset of field elements defines the polynomial $\prod_v (X + v)$, with one factor per occurrence of an element $v$.
A polynomial factors into linear terms in exactly one way, and two multisets are therefore equal exactly when their product polynomials are equal.

\fakeparagraph{Permutation Argument}The permutation argument proves that two polynomials take the same multiset of values over the evaluation domain.
Each polynomial's evaluations over $H$ form a multiset, and the product encoding turns the equality of the two multisets into an equality of product polynomials.
For $a(X)$ and $b(X)$ over $H$, the argument proves
\begin{equation*}
\prod_{i=0}^{n-1} \bigl(X + a(\omega^i)\bigr)
= \prod_{i=0}^{n-1} \bigl(X + b(\omega^i)\bigr) .
\end{equation*}
We invoke the argument through the interface $$\mathsf{Eq}(a(X), b(X)),$$ where each polynomial input stands for the multiset of its evaluations over the domain.

\fakeparagraph{Lookup Argument}The lookup argument proves that the set of values of one polynomial, the query, is contained in the set of values of a second, the table.
Containment is weaker than multiset equality, as a table entry may be used by several queries or by none, and the two products are therefore matched by reweighting the table side.
For a query polynomial $f(X)$ and a table polynomial $t(X)$ over $H$, each table factor is raised to a multiplicity $m_j \geq 0$, supplied by the prover as a witness subject to $\sum_j m_j = n$:
\begin{equation*}
\prod_{i=0}^{n-1} \bigl(X + f(\omega^i)\bigr)
= \prod_{j=0}^{n-1} \bigl(X + t(\omega^j)\bigr)^{m_j} .
\end{equation*}
The multiplicity $m_j$ counts the queries that use table entry $j$, unused entries take $m_j = 0$, and the identity holds for some choice of multiplicities exactly when every evaluation of $f(X)$ occurs among the evaluations of $t(X)$.
Fixing every multiplicity to one instead recovers the permutation argument~\cite{haboeck2022logup}.
We invoke the argument through the interface $$\mathsf{Mem}\bigl(f(X),\, t(X)\bigr).$$
\subsection{Extensions}
\label{sec:bg-extensions}
The inputs to the two arguments extend beyond single trace polynomials.
Both identities consume their inputs only through their evaluations, and any polynomial expression over the trace polynomials therefore defines a valid input, with its evaluations as the encoded multiset.
We define three such constructions, each extending the interfaces $\mathsf{Eq}$ and $\mathsf{Mem}$ with a notation for its inputs.

\fakeparagraph{Vectors}
A trace entry is often a tuple of values across several polynomials, while the product factors range over single field elements.
A verifier challenge $\alpha \xleftarrow{\$} \mathbb{F}$ compresses each tuple into one element, and two distinct tuples collide at a random $\alpha$ with negligible probability by the Schwartz--Zippel lemma.
Both interfaces accept vectors of polynomials $(f_1(X), \dots, f_d(X))$, entering the identity as the combination
\begin{equation*}
\vec{f}(X) = \sum_{c=1}^{d} \alpha^{c-1}\, f_c(X) .
\end{equation*}

\fakeparagraph{Selection}
A claim rarely concerns a polynomial's full domain, and an argument may apply to one region of the trace or to entries the prover designates.
Such restrictions are realized through \emph{selector} polynomials, which evaluate to $1$ on a subset of the domain and to $0$ elsewhere.
A selector is \emph{fixed} when preprocessed into the constraint system and \emph{witness} when supplied by the prover, and selectors compose by multiplication.
For a subset $S \subseteq H$, we write $f(X)|_S$ for the evaluations of $f(X)$ over $S$, entering the identity as the composite
\begin{equation*}
f(X)\big|_S = q_S(X) \cdot f(X) + \bigl(1 - q_S(X)\bigr) \cdot \delta ,
\end{equation*}
where $q_S(X)$ is the selector of $S$ and $\delta$ is a public default value.
\fakeparagraph{Unions}
A claim may collect values from several polynomials or several subsets at once, and the corresponding input is the union, in which each part contributes its evaluations.
Unioning multiplies the product polynomials, as the product over a union is the product over its parts:
\begin{equation*}
\prod_{v \,\in\, f(X)|_S \,\uplus\, g(X)|_{S'}} \bigl(X + v\bigr)
= \prod_{v \,\in\, f(X)|_S} \bigl(X + v\bigr) \cdot
\prod_{v \,\in\, g(X)|_{S'}} \bigl(X + v\bigr) .
\end{equation*}

%% file: sections/relatedwork.tex
\section{Related Work}
\label{sec:related}

\paragraph{Verifiable Semantic Search.}

Existing work on verifiable semantic search largely targets cluster-based indices.
VeriRAG~\cite{verirag} and V3DB~\cite{v3db} construct zero-knowledge proofs for search over IVF indices~\cite{videogoogle,ivfpq}, which partition the data into clusters offline and answer a query by exhaustively scanning the clusters nearest to it.
This regular structure encodes compactly as fixed polynomial constraints but lowers recall, as the clusters are selected by centroid distance alone, and for a query whose nearest neighbor lies outside them, the search returns a more distant vector. Our evaluation confirms this gap, as the product quantization these systems rely on bounds recall well below deployable targets (\cref{sec:evaluation}).

Concurrent to our work, Jiang et al.\ propose zkRAG~\cite{zkrag}, a zero-knowledge proof that directly encodes the HNSW search algorithm, differing from \oursystem' restructuring in two respects.
First, zkRAG preserves the layered graph structure and therefore relies on revealing the number of steps taken at every layer.
\oursystem' merged traversal runs under a single fixed bound (\cref{sec:unified-greedy}) and reveals no query-dependent information beyond the search result.
Second, zkRAG verifies the two priority queues of beam search through a specialized PIOP for queue updates.
\oursystem replaces the two queues with a single bounded set, which we show leaves the search result unchanged (\cref{sec:single-set-beam}), and efficiently constrains the set's updates through one batched merge per iteration.

\fakeparagraph{Private ANN Search}
Private ANN search addresses a setting orthogonal to ours: these systems hide the query from the server but assume it executes the search faithfully.
Tiptoe supports cluster-based search using linearly homomorphic
encryption~\cite{tiptoe}, HERS performs exhaustive matching over FHE-encrypted representations~\cite{hers}, and Compass performs HNSW search while hiding access
patterns with ORAM~\cite{compass}. 
Compass additionally provides integrity, but under an inverted trust model: the client owns the database, maintains the index graph itself, and outsources only storage, whereas in our setting the provider controls both and the client learns only results.
Onyx extends the ORAM-based design by moving the client logic into a server-side TEE, executing the search inside an enclave~\cite{onyx}, obtaining integrity from attestation at the cost of hardware trust.

\paragraph{Verifiable ML Inference.}
A recent line of work on verifiable inference, often referred to as zkML, uses zero-knowledge proofs to show that a committed model produced a claimed output~\cite{zkml, artemis, zkllm, zkgpt}. These works are complementary to ours and can be combined to build RAG pipelines. Verifiable inference frameworks can be used to prove the model computations while
\oursystem provides verifiability for the retrieval step.
Together, they enable a client to verify an end-to-end retrieval-augmented response. 

%% file: sections/overview.tex
\section{System Overview}\label{sec:overview}

\oursystem is a system that enables a client to verify that a semantic search query was answered correctly against a provider-maintained database. This section introduces the setting, the threat model and security goals, and describes our approach at a high level. 
The full construction follows in \S\ref{sec:design}.

\paragraph{Setting.}
We consider a standard semantic search service, in which a provider maintains a database $D$ of $N$ vectors. 
Given an input query $q \in \mathbb{R}^d$, the service aims to find the $k$ nearest vectors under a distance measure \di{\cdot}{\cdot} such as the $\ell_2$ distance.
To scale to large databases, the provider builds an HNSW graph $G$, which organizes the dataset into $L+1$ layers $G_0, \ldots, G_L$, where $G_0$ holds the full dataset and each $G_\ell$ contains a small fraction of the nodes of $G_{\ell-1}$.
To resolve a query $q$, the provider runs \textsc{HNSW-Search} (\Cref{alg:hnsw_search}): the search starts at a fixed entry point in the top layer, greedily visiting neighboring nodes and proceeding downward when no neighbor at the current layer is closer to $q$. 
In the bottom layer, it performs a beam search of width $\mathit{ef}$ and returns $\mathrm{top}_k(W)$, the $k$ elements of the final candidate set $W$ nearest to the query, as the result.
Throughout, $N(u)$ denotes the neighborhood of node $u$ in the graph.

\paragraph{System Goals.}
A verifiable semantic search system must achieve three properties, (i)~\emph{privacy}, a client learns nothing beyond what the query result reveals, (ii)~\emph{soundness}, a client receives the correct search result with respect to the committed index, and (iii)~\emph{efficiency}, the per-query cost to both provider and client is sublinear in the database size.
We state soundness relative to the index rather than the database, as approximate search is inherently index-dependent and a different index can yield a different result set.

\paragraph{Threat Model.}
We consider two parties, a provider that holds the index and answers queries, and a client that issues queries and verifies the responses.
Either party may be corrupted by a malicious adversary.
A corrupted provider may deviate from the protocol arbitrarily:
it may answer from a stale or altered index, truncate the traversal, bias the
ranking, or forge a proof for a result. 
A corrupted client may try to learn as much information as possible from the provider's response, such as information about the index graph's structure, edges, or embeddings.
The query is sent to the provider in the clear.
We assume that before serving any query, the provider publishes a binding commitment to its index $G$ that clients can learn.
Finally, the search result is bound to the committed index, but claims about the well-formedness of the index are orthogonal\footnote{Existing techniques can be used to show the well-formedness, i.e., that its edges connect genuinely nearby vectors, or that its construction followed any particular distribution.}.

\subsection{Our Approach}

At the core of \oursystem is a new zero-knowledge proof construction for the \textsc{HNSW-Search} algorithm. 
Unlike many graph algorithms, including shortest-path search, HNSW admits no succinct certificate that can be verified in place of recomputing the result.
Its traversal is greedy and data-dependent, and the correctness of a result can therefore be established only by re-executing the search itself.
Our proof consequently verifies a complete execution of the search, including the visited nodes, their distance comparisons, and the intermediate algorithm state.
More formally, our proof shows that a set of embeddings $R$ is the result of correctly executing the \textsc{HNSW-Search} algorithm for a query $q$ on a private index $G$ committed to by $c$, i.e., a proof of knowledge for the relation 
\begin{equation*}
	\mathcal{R} = \left\{
	\begin{array}{c}
    ((c, q, R), G) : \\
    \begin{array}{l}
    R = \textsc{HNSW-Search}_{\mathit{ef}, k}(G, q) \\
    \wedge\; c = \mathsf{Commit}(G)
    \end{array}
    \end{array}
  \right\},
\end{equation*}
where $\mathit{ef}$ and $k$ are public configuration parameters of HNSW.

\fakeparagraph{Setup Phase}
The prover commits to the index $G$ with three commitments: one to an embedding table, which stores the embeddings of the entire database $D$, and two to edge tables, one for the upper $L$ layers of the HNSW graph and one for the bottom layer. 
The $i$-th row of the embedding table contains a node identifier and its embedding vector $e_i \in \mathbb{F}_p^{D}$, i.e.,
$$
\left[id_i, e_i^{(0)},\ldots,e_i^{(D-1)}\right].
$$
We split the edge list into two tables, $\Ted^{\mathrm{greedy}}$ and $\Ted^{\mathrm{beam}}$, as nodes in the upper $L$ layers have $M$ neighbors and nodes in layer $0$ have $2M$.
Row $i$ in the greedy table $\Ted^{\mathrm{greedy}}$ contains a tuple of node identifier $id_i$ and level $\ell_i$, followed by $M$ tuples of neighbors, i.e.,
$$
\left[id_i, \ell_i, id_{nb}^{(0)},\ell_i,\ldots,id_{nb}^{(M-1)},\ell_i, id_{i},\ell_{i}-1\right].
$$
The final tuple $(id_{i},\ell_{i}-1)$ represents a downward edge to the same node one layer below.
This flattening of the upper layers into a single graph $G_\mathrm{upper}$ enables the optimizations of the search's dynamic termination that we present in \S\ref{sec:restructuring}.

\paragraph{Online Phase: Consistency.}
In the online phase, the prover shows that every edge and embedding on the traversed path is consistent with the commitments to $G$.
Lookup arguments are the natural primitive for these consistency checks, as they prove that claimed values appear in a committed table.
Most, however, incur prover cost linear in the table size, so the proof costs as much as a full dataset scan and the sublinearity that motivates the index is lost.
We therefore specifically choose the cq lookup argument~\cite{cq, lookupsok}, whose per-query prover cost is independent of the table size after a one-time preprocessing of the table.
We perform this preprocessing once at index construction, in time $O(N \log N)$ matching the cost of constructing $G$, and the prover's per-query cost then scales with the traversal's length rather than the graph's size.

\paragraph{Online Phase: Traversal.}
We encode HNSW's traversal in two parts, the greedy search in the upper layers and the beam search in the bottom layer.
In the flattened upper graph, each node is referenced by a tuple $(u, \ell)$ of node identifier and level, as the same node appears at several layers.
In place of each layer's data-dependent termination condition, we bound the traversal by a fixed step budget $T_g$ over the unified graph $G_\mathrm{upper}$ and $T_b$ for the beam search in the lowest layer.
For the greedy search, the prover shows at each step that the selected node is a neighbor of the current node and the nearest one to the query, with all steps proven jointly by a single batched argument (\S\ref{sec:greedy-constraints}).
Once the traversal reaches the bottom layer, it no longer suffices to identify a single next node, as the beam search must maintain and update a dynamic set of $\mathit{ef}$ candidates.
Encoding this dynamic behavior is the main challenge of our construction, and we address it in \S\ref{sec:restructuring}.

\begin{algorithm}[t]
\small
\setstretch{0.92} 
\SetKwInOut{Input}{Input}
\SetKwInOut{Output}{Output}
\DontPrintSemicolon
\caption{$\textsc{HNSW-Search}_{\mathit{ef}, k}(G, q)$}
\label{alg:hnsw_search}
\Input{Multi-layer graph $G$ with layers $G_0, \ldots, G_L$ and entry point $\mathit{ep}_{G_L} \in G_L$, query vector $q$, beam width $\mathit{ef}$, number of results $k$}
\Output{$k$ nearest neighbors of $q$ found by the search}
$\mathit{ep} \leftarrow \mathit{ep}_{G_L}$\;\label{ln:epset}
\For{$\ell \leftarrow L$ \KwTo $1$}{
    $W \leftarrow \textsc{Search-Layer}_{1}(G_\ell, q, \mathit{ep})$\;\label{ln:searchlayercall}
    $\mathit{ep} \leftarrow \arg\min_{x \in W} \mathit{dist}(x, q)$\;\label{ln:eplast}
}
\BlankLine
$W \leftarrow \textsc{Search-Layer}_{\mathit{ef}}(G_0, q, \mathit{ep})$\;
\BlankLine
\Return{$\mathrm{top}_{k}(W)$}
\end{algorithm}

\begin{algorithm}[t]
\small
\setstretch{0.92}
\SetKwInOut{Input}{Input}
\SetKwInOut{Output}{Output}
\DontPrintSemicolon
\caption{$\textsc{Search-Layer}_{\mathit{ef}}(G_\ell, q, \mathit{ep})$}
\label{alg:search_layer}
\Input{layer graph $G_\ell$, query vector $q$, entry point set $\mathit{ep}$, beam width $\mathit{ef}$}
\Output{$\mathit{ef}$ nearest neighbors of $q$ found at this layer}
$v \leftarrow \mathit{ep}$;\ $C \leftarrow \mathit{ep}$;\ $W \leftarrow \mathit{ep}$\label{ln:init}\;
\While{$|C| > 0$}{\label{ln:while}
    $c \leftarrow \arg\min_{x \in C} \mathit{dist}(x, q)$;\ $C \leftarrow C \setminus \{c\}$\label{ln:pop}\;
    $f \leftarrow \arg\max_{x \in W} \mathit{dist}(x, q)$\;
    \If{$\mathit{dist}(c, q) > \mathit{dist}(f, q)$}{
        \textbf{break}\label{ln:break}\;
    }
    \For{$e \in N_{G_\ell}(c)$\label{ln:forneighbor}}{
        \If{$e \notin v$\label{ln:visitcheck}}{
            $v \leftarrow v \cup \{e\}$\label{ln:addvisited}\;
            $f \leftarrow \arg\max_{x \in W} \mathit{dist}(x, q)$\;
            \If{$\mathit{dist}(e, q) < \mathit{dist}(f, q)$ \textnormal{or} $|W| < \mathit{ef}$\label{ln:distcheck}}{
                $C \leftarrow C \cup \{e\}$;\ $W \leftarrow W \cup \{e\}$\label{ln:insert}\;
                \If{$|W| > \mathit{ef}$}{
                    $W \leftarrow W \setminus \{\arg\max_{x \in W} \mathit{dist}(x, q)\}$\label{ln:evict}\;
                }
            }
        }
    }
}
\Return{$W$}
\end{algorithm}

%% file: sections/design.tex
\section{Design}
\label{sec:design}
Proving HNSW search in a zkSNARK requires bridging a mismatch between the algorithm's data-dependent execution and the fixed constraint system that a zkSNARK demands.
How long the search runs and how much state it accumulates both vary with the query, and the constraint system must provision for both in advance.
We first identify the two sources of this mismatch, then resolve it in two stages.
The first stage (\S\ref{sec:restructuring}) is algorithmic and confines the data dependence that is not inherent to the search.
The second stage (\S\ref{sec:techniques}) is algebraic and encodes what remains as compact polynomial constraints.
\fakeparagraph{Challenge 1: Data-dependent Termination}
HNSW search is inherently dynamic.
While on average the search terminates within $\log N$ steps, the actual number of search steps at each layer varies and is not known in advance.
The \textsc{Search-Layer} algorithm (\Cref{alg:search_layer}) returns when the nearest unexplored node is further from the query than the furthest node in the result set (line~\ref{ln:break}).
This termination condition depends on the query and the index, and a fixed constraint system must therefore encode the worst-case traversal of every layer.
The proving cost of every query is then that of $L$ worst-case layer traversals.
We collapse the upper $L$ layers into a unified graph representation $G_\mathrm{upper}$, turning their $L$ traversals into one.
We then bound the total number of steps of the greedy and beam search phases, in place of encoding $L$ worst-case layer traversals.
\fakeparagraph{Challenge 2: Beam Search State Tracking}
Where greedy search maintains only a single node as its frontier,
beam search holds up to $\mathit{ef}$ nodes at the same time.
As a result, the choice of which nodes to explore next depends not only on the local neighborhood of the current node, but on all $\mathit{ef}-1$ other candidates as well: a node may be dropped from consideration because closer candidates exist in an entirely different region of the graph.
The algorithm tracks this exploration with three data structures (line~\ref{ln:init}): a result set $W$ of at most $\mathit{ef}$ elements, and a candidate set $C$ and visited set $v$ that both grow with the traversal (lines~\ref{ln:addvisited},~\ref{ln:insert}). 
Each neighbor of an expanded node is first checked against the visited set $v$ (line~\ref{ln:visitcheck}) and then admitted into both $W$ and the candidate queue $C$ if it improves on $W$'s furthest element.
The admission depends on the evolving state, as each admission can evict $W$'s furthest element (line~\ref{ln:evict}) and tighten the threshold for every neighbor after it.
A direct proof would therefore have to unroll the loop into one step per neighbor, each carrying its own copy of $W$, $C$, and $v$ and encoding every branch the update could take.
With this, the cost of expanding a single node would scale as a product of the neighborhood size and the worst-case queue sizes.
Through our restructuring of the algorithm state, we process each neighborhood in a single batched update instead of a chain of sequential insertions.
 
\subsection{Restructuring HNSW} 
\label{sec:restructuring}

In this subsection, we reformulate the HNSW algorithm to support more efficient dynamic termination and reduce the complexity of proving beam search.
Concretely, we fix the number of search steps, resolving the data-dependent termination of Challenge~1, and collapse the interdependent sets of Challenge~2 into a single set of fixed size $\mathit{ef}$.
First, we organize the $L+1$ layers into two graphs, merging the $L$ per-layer greedy traversals into a single traversal by encoding layer transitions as edges, removing the need for a separate encoding per layer.
Second, we replace the beam phase algorithm with an equivalent procedure whose internal state is only a single bounded set, turning each iteration into one batched update rather than a chain of sequential insertions. 
We prove that the restructured search (\Cref{alg:restructured}) produces the same output as the original (\Cref{alg:hnsw_search}).

\subsubsection{Unifying Greedy Search}
\label{sec:unified-greedy}
To avoid padding every layer's traversal to its worst-case length (Challenge~1), we reorganize the multi-layer graph $G$ into two components, the bottom graph $G_0$ and an upper graph $G_\mathrm{upper}$ that merges the $L$ layers of the greedy search.
The starting point of the search in each layer is the end node of the previous layer, and the downward transition can therefore be represented as an edge from $(v, \ell)$ to $(v, \ell{-}1)$, making the layer boundary disappear into the edge set.
This enables us to replace the $L$ separate calls to $\textsc{Search-Layer}_1(G_\ell, q, \mathit{ep})$ (line~\ref{ln:searchlayercall}) with a single call to $\textsc{Search-Layer}'_1(G_\mathrm{upper}, q, \mathit{ep})$.
$\textsc{Search-Layer}'$ is defined in the same way as $\textsc{Search-Layer}$ but compares nodes by the pair $(d(v, q), \ell)$ lexicographically, i.e., the walk moves to a strictly closer neighbor whenever one exists and otherwise breaks distance ties in favor of the lower layer.
We define the unified upper graph $G_{\mathrm{upper}}$ as the union of all per-layer edge sets together with the drop-down edges $\{((v, \ell), (v, \ell{-}1)) : v \in G_\ell,\ \ell \geq 1\}$.
The single invocation $\textsc{Search-Layer}'_{1}(G_\mathrm{upper}, q, \{(\mathit{ep}, L)\})$ terminates at $(v^*_1, 0)$, the node reached by the original greedy phase (lines~\ref{ln:epset}--\ref{ln:eplast} of \Cref{alg:hnsw_search}), after $T_g = \sum_{\ell=1}^{L} T^*_\ell + L$ steps, where $T^*_\ell$ counts the greedy steps at layer $\ell$.
We state and prove this equivalence as \Cref{thm:greedy-equivalence} in \Cref{app:proof-greedy}.
Since $(v, \ell)$ and $(v, \ell{-}1)$ share the same embedding, $d((v, \ell{-}1), q) = d((v, \ell), q)$, and the lexicographic comparison therefore selects the drop-down edge only when no intra-layer neighbor is strictly closer to $q$, i.e., exactly when $\textsc{Search-Layer}_1(G_\ell, q, \cdot)$ would terminate.
The walk on $G_\mathrm{upper}$ thus reproduces each per-layer search and descends at its local minimum.

We pad each of the two remaining traversals to a fixed budget, the walk in $G_\mathrm{upper}$ to $T_g$ steps and the search in $G_0$ to $T_b$ steps.
The restructuring leaves only these two bounds to provision, in place of a worst-case traversal for each of the $L$ layers, and both are set from the observed distribution of total step counts.
A search exceeding its budget terminates early with its current candidate set, which at 95th-percentile budgets costs less than one recall@1 point on the integer datasets (\S\ref{sec:evaluation}).

\subsubsection{Single-set Beam Phase}
\label{sec:single-set-beam}
Beam search is more complex than greedy because the algorithm maintains multiple candidates simultaneously, tracked by two priority queues $W$ and $C$, and each neighbor must be evaluated against the evolving state of both.
Inserting a node into the result set $W$ can change its worst element and, with it, the insertion threshold for the next node.
The $2M$ neighbors of each selected node must therefore be processed sequentially, requiring $O(M \cdot \mathit{ef})$ constraints per iteration.
The candidate set $C$ compounds this issue, as its size has no fixed bound and must be allocated for the worst case.
We eliminate both obstacles by reformulating the beam search around a single bounded set that can be updated efficiently, with a flag marking the elements already processed.
The reformulation rests on two observations:
(i) the state of $W$ after processing a neighborhood does not depend on the order in which neighbors are inserted, and
(ii) although the contents of $C$ do depend on insertion order, this difference never influences the search result.
Using the two observations, we replace the sequential neighbor insertions of each step with a single merge of the $2M$ neighbors into the bounded set, and the queue $C$ with a flag.
In the following, we discuss and prove both observations, then present the restructured beam search.

\paragraph{Order Independence of the Result Set.}
At each iteration of \textsc{Search-Layer} (\Cref{alg:search_layer}), the neighborhood exploration loop (line~\ref{ln:forneighbor}) inserts neighbors into $W$ sequentially, evicting the furthest element when $|W| > \mathit{ef}$ (line~\ref{ln:evict}).
Although the intermediate states of $W$ vary with the insertion order, its final state is always $\mathrm{top}_{\mathit{ef}}(W \cup \text{neighbors})$, the $\mathit{ef}$ elements nearest the query among $W$ and the neighbors combined.
The neighborhood can therefore be processed as a single batched operation, whose constraint cost scales additively ($O(M + \mathit{ef})$) rather than multiplicatively ($O(M \cdot \mathit{ef})$).

\begin{lemma}[Order Independence of $W$]
\label{lem:order-independence}
Let $W$ be a set of at most $\mathit{ef}$ elements, let $A$ be a set of new elements, and let $q$ be a query point.
Inserting the elements of $A$ into $W$ in any order, with the furthest element from $q$ evicted whenever $|W| > \mathit{ef}$, as in lines~\ref{ln:forneighbor}--\ref{ln:evict} of \Cref{alg:search_layer}, always yields $\mathrm{top}_{\mathit{ef}}(W \cup A)$, the $\mathit{ef}$ elements of $W \cup A$ nearest to $q$.
\end{lemma}

\paragraph{Invariance of the search result under $C$.}
The order-independence of $W$ does not immediately imply the same property for the candidate set $C$.
A neighbor enters $C$ only if it is close enough to be in $W$ at the moment it is processed, and the admission threshold, namely the
distance of $W$'s current worst element, changes as the
neighborhood is processed.
Different insertion orders therefore expose each neighbor to
different thresholds, and a node admitted under one order may
be rejected under another.
For example, let $W$ have capacity $\mathit{ef} = 2$ and hold
nodes with distances $\{5, 8\}$, and let the current node's neighbors be represented by distances
$\{ 3, 7, 9 \}$.
If 3 is processed first, it evicts 8 and the threshold drops
to 5, so 7 is rejected and only 3 is added to $C$.
If 7 is processed first, it is added to $C$, only to be evicted from $W$ when 3 arrives.
Both orders yield the same final $W$, holding $\{3, 5\}$, but
different candidate sets: $C = \{3\}$ versus $C = \{3, 7\}$.

However, this difference never results in a different output of \textsc{Search-Layer}.
First, any node for which the two executions disagree is only included in $W$ temporarily and is subsequently removed; therefore, its distance is greater than that of every node in the final set $W$.
Second, $C$ is kept in ascending order of distance, and \textsc{Search-Layer} terminates as soon as the extracted node is further than $W$'s worst element (line~\ref{ln:break}).
A disputed node is therefore selected only after every closer
candidate, and when its turn comes, it
triggers termination rather than expansion.
The extra elements of $C$ are redundant, as both executions
process the same nodes in the same order.
 
\begin{lemma}[Candidate-Selection Invariance]
\label{lem:pop-invariance}
Two executions of \textsc{Search-Layer} (\Cref{alg:search_layer}) on the same $(G, q, \mathit{ep}, \mathit{ef})$, that differ only in the order in which neighbors are inserted at line~\ref{ln:forneighbor}, extract the same node from $C$ at line~\ref{ln:pop} at every iteration.
\end{lemma}
\noindent
We prove \Cref{lem:order-independence,lem:pop-invariance} in \Cref{app:proof-order,app:proof-invariance}.

The two lemmas show that the two-queue structure is unnecessary: neither the sequential insertion order nor the unbounded queue $C$ affects the search result.
The result set $W$'s final state is order-independent, and $C$ is only required to track which candidates in $W$ have not yet been processed.
As a result, we can replace $C$ with a flag on the elements of $W$ that indicates whether they have been processed (Phase~2 of \Cref{alg:restructured}).
At each iteration, the nearest unprocessed element is selected, and its entire neighborhood is merged into the set in a single batch by taking the union of the current candidates and the selected node's neighbors, sorting by distance, and keeping the nearest $\mathit{ef}$.
We provide a formal description of the restructured version in~\Cref{alg:restructured}.
The sequential per-neighbor insertions are replaced by one batched merge per step, and the unbounded queue $C$ is eliminated entirely.

\begin{theorem}[Beam Search Equivalence]
\label{thm:beam-equivalence}
Let $\mathit{ep} \in G_0$, $q \in \mathbb{R}^d$, and $\mathit{ef}$ be an entry
point, query, and beam width, and let $T^*$ be the number of iterations of
$\textsc{Search-Layer}_{\mathit{ef}}(G_0, q, \mathit{ep})$
(\Cref{alg:search_layer}).
For any budget $T_b \geq T^*$, Phase~2 of \Cref{alg:restructured} returns the
same final set.
\end{theorem}
\noindent We prove \Cref{thm:beam-equivalence} in \Cref{app:proof-beam}.

\begin{algorithm}
\SetKwInOut{Input}{Input}
\SetKwInOut{Output}{Output}
\SetKwComment{Comment}{// }{}
\DontPrintSemicolon
\caption{$\textsc{HNSW-Search-Restruct}_{\mathit{ef}, k, T_g, T_b}(G, q)$}
\label{alg:restructured}
\Input{graph $G$ with unified upper graph $G_{\mathrm{upper}}$,
  layer-$0$ graph $G_0$, and entry point $\mathit{ep}_{G_L}$, query vector $q$, beam width $\mathit{ef}$, number of results $k$, step budgets $T_g$ and $T_b$}
\Output{$k$ nearest neighbors of $q$ found by the search}
\BlankLine
\Comment{Phase 1: greedy search on $G_{\mathrm{upper}}$}
$c \leftarrow (\mathit{ep}_{G_L}, L)$\;
\For{$t \leftarrow 1$ \KwTo $T_g$}{
    $c \leftarrow \arg\min_{n \in N_{G_{\mathrm{upper}}}(c)} \mathit{dist}(n, q)$\;
}
$(\mathit{ep}_0, \cdot) \leftarrow c$\;
\BlankLine
\Comment{Phase 2: batched beam search on $G_0$}
$W \leftarrow \{(\mathit{ep}_0,\, \mathit{dist}(\mathit{ep}_0, q),\, 0)\}$\;
$\mathit{visited} \leftarrow \{\mathit{ep}_0\}$\;
\For{$t \leftarrow 1$ \KwTo $T_b$}{
    \If{$\{(n, d, p) \in W : p = 0\} = \emptyset$}{\textbf{continue}\;}
    $(n_{\mathrm{sel}}, d_{\mathrm{sel}}, \cdot) \leftarrow \arg\min_{(n, d, p) \in W,\; p = 0}\, d$\;
    $W \leftarrow W \setminus \{(n_{\mathrm{sel}}, d_{\mathrm{sel}}, 0)\} \cup \{(n_{\mathrm{sel}}, d_{\mathrm{sel}}, 1)\}$\;
    $S \leftarrow W \cup \{(e,\, \mathit{dist}(e, q),\, 0) : e \in N_{G_0}(n_{\mathrm{sel}}) \setminus \mathit{visited}\}$\;
    $\mathit{visited} \leftarrow \mathit{visited} \cup N_{G_0}(n_{\mathrm{sel}})$\;
    $W \leftarrow \mathrm{top}_{\mathit{ef}}(S)$\;
}
\BlankLine
\Return{$\mathrm{top}_{k}(W)$}
\end{algorithm}

\subsection{Constraint Design}
\label{sec:techniques}

The restructured algorithm of~\Cref{sec:restructuring} makes each traversal step a single batched operation, which is significantly more efficient to prove.
A remaining difficulty is the search's branching, which depends on which nodes have been considered and which candidates have been processed.
The algorithm tracks this state in sets of query-dependent size and contents.
In this section, we express the search as polynomial constraints, with this state supplied by the prover as witness flags, so each branch reduces to a value substitution, every timestep runs identical constraints, and neither set is ever materialized.
We present the constraints of the greedy search first, as the simpler of the search's two phases.
We then discuss beam search as an extension of the greedy constraints.

\subsubsection{Greedy Search}
\label{sec:greedy-constraints}
\fakeparagraph{Trace Encoding}
The greedy search traverses the unified upper graph of \S\ref{sec:restructuring}, at each step moving from the current node to its neighbor nearest to the query, for a fixed budget of $T_g$ steps.
Concretely, the trace is encoded in witness polynomials over a shared evaluation domain, and the walk occupies $T_g$ regions of $B = M + 1$ consecutive points, one region per step.
The first point of region $t$ is the \emph{selected point}, at which the polynomials encode the node selected at step $t$, and the remaining $M$ points are the \emph{neighbor points}, encoding the selected node's neighborhood in the graph.
The selected points across the segment form the position set $\mathsf{Sel}_g$ and the neighbor points the set $\mathsf{Nb}_g$, each with a preprocessed selector polynomial, $q_{\mathsf{Sel}_g}$ and $q_{\mathsf{Nb}_g}$.
For a position set $S$, we write $S^{(t)}$ for its positions in region $t$, and $S^{-}$ for $S$ without the last point of each region.
The polynomials $\polyid(X)$, $\polyl(X)$, and $\polyd(X)$ encode the identifier, layer, and claimed squared distance to the query of the node at every point, and $\polye_{1}(X), \dots, \polye_{D}(X)$ its embedding.
Each step of the walk explores the selected node's neighborhood, computes every node's distance to the query, and selects the nearest neighbor.
We prove the sequence with three constraint sets, \emph{edge consistency} for the explored neighborhoods, \emph{distance verification} for the computed distances, and \emph{candidate selection} for the nearest neighbor.

\fakeparagraph{Edge Consistency}
The neighbor points of each region encode the selected node's neighborhood, supplied by the prover as evaluations of $\polyid(X)$ and $\polyl(X)$.
Each claimed neighborhood must therefore be proven equal to the selected node's edge list in the committed graph.
The edge table $\Ted$ stores a node's full edge list as a single entry, while the trace encodes the list across the region's $M$ neighbor points.
We assemble each region's edge list into one tuple through replicated polynomials $\polyidrep{j}(X)$ and $\polylrep{j}(X)$ for $j \in [M]$, copy-constrained at every point of a region to the values of $\polyid(X)$ and $\polyl(X)$ at the region's $j$-th neighbor point, and write $\vec{a}(X)$ for the assembled vector $(\polyid(X),\, \polyl(X),\, \polyidrep{1}(X),\, \polylrep{1}(X), \dots, \polyidrep{M}(X),\, \polylrep{M}(X))$.
A single membership argument enforces that every region's assembled list is an entry of the edge table:
\[
  \Mem\bigl( \restr{\vec{a}(X)}{\Selg},\; \Ted \bigr) .
\]
Throughout, $\restr{f(X)}{S}$ denotes the restriction of $f(X)$ to the position set $S$, as defined in \S\ref{sec:bg-extensions}.

\fakeparagraph{Distance Verification}
Every node's distance to the query is supplied by the prover as an evaluation of $\polyd(X)$, and each claimed value must be proven equal to the true distance.
A node's distance is determined by its embedding, and correctness is enforced in two parts, an embedding consistency lookup tying each node's embedding to the committed dataset and a polynomial identity tying each claimed distance to the embedding.
The lookup proves every point's identifier and embedding to be an entry of the embedding table $\Temb$:
\[
  \Mem\bigl( \restr{(\polyid(X),\, \polye_{1}(X), \dots, \polye_{D}(X))}{\Selg \cup \Nbg},\; \Temb \bigr) .
\]

\noindent The identity
\begin{equation*}
\Bigl( \polyd(X) - \sum_{i=1}^{D} \bigl( q_i - \polye_{i}(X) \bigr)^2 \Bigr) \Big|_{\Nbg} = 0
\end{equation*}
constrains each claimed distance to the squared $\ell_2$ distance between the public query $(q_1, \dots, q_D)$ and the proven embedding.

\fakeparagraph{Candidate Selection}
Selecting the nearest neighbor amounts to two claims, that the selected node is drawn from the preceding step's neighborhood and that no neighbor is nearer to the query.
We enforce the first claim through a membership constraint, proving that the selected node of each step is a member of the preceding step's neighborhood.
Concretely, both are encoded in $\polyid(X)$: the selected node at region $t + 1$'s selected point and the neighborhood at region $t$'s neighbor points.
We instantiate a membership argument of $\polyid(X)$ into itself, with the selected positions as the query and the neighbor positions as the table:
\begin{equation}
\label{eq:membership}
  \Mem\bigl( \restr{\polyid(X)}{\Selg^{(t+1)}},\; \restr{\polyid(X)}{\Nbg^{(t)}} \bigr) .
\end{equation}

The argument enforces that the identifier at each selected point equals one of those at the preceding region's neighbor points.
To constrain the selected node to be the nearest to the query among all neighbors, it suffices to show that the difference between each neighbor's distance and the selected node's distance is nonnegative.
We introduce a helper polynomial $\polydsel(X)$, equal at every point of region $t$ to the selected distance of region $t + 1$ through a copy constraint.
A membership argument into the range table $\Trange$, a preprocessed table of the nonnegative integers $\{0, \dots, 2^{r} - 1\}$, enforces that every difference is nonnegative:
\begin{equation}
\label{eq:minimality}
  \Mem\bigl( \restr{(\polyd(X) - \polydsel(X))}{\Nbg},\; \Trange \bigr) .
\end{equation}
\fakeparagraph{Timestep-Tagged Argument}
The membership arguments of candidate selection are instantiated once per timestep, as each is restricted to a different pair of position sets, $\Selg^{(t+1)}$ against $\Nbg^{(t)}$.
The $T_g$ separate instances result in significant prover overhead, each requiring auxiliary polynomials and identities of its own.
The pattern appears throughout the rest of the constraint design, as membership and equality claims are typically local to a timestep, each naively requiring a separate argument restricted to its region.
We reduce this overhead by collapsing the per-timestep instances of each claim into a single argument over the full trace, which we refer to as a \emph{timestep-tagged argument}.
Concretely, we extend every tuple with a shared preprocessed \emph{tag polynomial} $\tau(X)$, equal to $t$ over region $t$, and prove the single argument over the tag-extended tuples.
Tuples match only when their tags agree, and a query therefore matches only the table entries of its own timestep.
For the neighborhood membership, the tagged argument proves all $T_g$ claims at once:
\[
  \Mem\bigl( \restr{(\polytag(X) - 1,\, \polyid(X))}{\Selg},\; \restr{(\polytag(X),\, \polyid(X))}{\Nbg} \bigr) .
\]
The query side's tag is shifted by one, tagging the selected point of region $t + 1$ with $t$ to group each selected node with the neighborhood it was selected from.

\subsubsection{Beam Search}
\label{sec:beam-constraints}
\fakeparagraph{Trace Encoding}
The beam phase operates on the single candidate set of \S\ref{sec:restructuring}, at each step processing the nearest unprocessed candidate and merging its neighborhood into the set, for a fixed budget of $T_b$ steps (Phase~2 of \Cref{alg:restructured}).
Concretely, the search is laid out across $T_b$ regions of $B_b = \max(\mathit{ef}, 2M) + 1$ consecutive evaluation points, one region per step.
A region must hold the step's candidate set of $\mathit{ef}$ entries and the selected candidate's neighborhood of $2M$ entries, and $B_b$ therefore fits the longer of the two alongside the selected point.
As in the greedy segment, the first point of region $t$ is the selected point, encoding the candidate processed at step $t$, with its identifier and distance carried by $\polyid(X)$ and $\polyd(X)$.
The two lists share the region's points through separate pairs of polynomials.
The polynomials $\polycan(X)$ and $\polydcan(X)$ encode one candidate's identifier and distance per point over the first $\mathit{ef}$ points, and $\polynb(X)$ and $\polydnb(X)$ one neighbor's identifier and distance per point over the first $2M$ points.
The candidate positions across the segment form the position set $\mathsf{Can}_b$, the neighbor positions $\mathsf{Nb}_b$, and the selected points $\mathsf{Sel}_b$, each with a preprocessed selector polynomial.
We discuss the constraints that enforce the correctness of the two operations, \emph{candidate selection} and \emph{candidate set update}.
\fakeparagraph{Candidate Selection}
We reuse the nearest-neighbor selection constraints of the greedy search to select the next candidate.
In particular, the membership argument (\Cref{eq:membership}) applies with the selected point as the query and the region's candidate entries as the table, proving the selected candidate a member of the step's candidate set, and the minimality check (\Cref{eq:minimality}) applies with $\polycan(X)$ and $\polydcan(X)$ in place of $\polyid(X)$ and $\polyd(X)$.
The difference is that the choice now depends on the \textsf{processed} status, as only unprocessed candidates must be selected.
We therefore extend the constraints to exclude processed candidates from the selection.
The prover supplies the processed flag as a polynomial $\polyproc(X)$, equal to $0$ at candidates not yet processed and $1$ at processed candidates.
We use the flag to substitute $\dinf$ for every processed candidate's distance, defining the substituted distance $\polydcansub(X)$:
\[
  \bigl( \polydcansub(X) - \polydcan(X) \cdot (1 - \polyproc(X)) - \dinf \cdot \polyproc(X) \bigr) \big|_{\Canb} = 0 ,
\]
where $\dinf$ is a public constant that exceeds every achievable squared $\ell_2$ distance.
The minimality check then applies with $\polydcansub(X)$ in place of $\polydcan(X)$:
\[
  \Mem\bigl( \restr{(\polydcansub(X) - \polydsel(X))}{\Canb},\; \Trange \bigr) .
\]
Every processed candidate's distance thus exceeds every unprocessed one's, and processed candidates never influence the selection.
The flag itself remains an unconstrained witness, and the last remaining step is enforcing its consistency with the trace, which we describe in \S\ref{sec:flag-verification}.
\fakeparagraph{Candidate Set Update}%
Following the restructuring of \S\ref{sec:restructuring}, the candidate set update merges the selected candidate's entire neighborhood into the candidate set in one batched operation.
Concretely, at each step $t$, the two lists are merged and sorted by distance, the $\mathit{ef}$ nearest entries become the candidate set of step $t + 1$, and the remaining entries are discarded.
We enforce the correctness of the merge, sort, and truncation through three respective arguments: an \emph{equality argument} enforcing that the input candidate and neighbor lists contain the same entries as the output candidate and discarded lists, a \emph{sorting argument} enforcing that every region's candidate entries are in ascending order of distance, and a \emph{truncation argument} enforcing that every discarded entry is at least as far from the query as the furthest retained one.

On the input side, the equality argument takes region $t$'s candidate and neighbor tuples.
On the output side, it takes region $t + 1$'s candidate entries together with region $t$'s discarded entries.
The retained entries live in the candidate polynomials themselves, as the candidate entries of region $t + 1$, while the discarded entries are encoded by two further polynomials, $\polyidrest(X)$ and $\polydrest(X)$, over the neighbor positions $\mathsf{Nb}_b$.
Applying timestep tagging (\S\ref{sec:greedy-constraints}), we tag the output side's candidate entries with $\polytag(X) - 1$, forcing the input tuples of step $t$ to match the candidate entries of region $t + 1$:
\[
\begin{aligned}
  \Eq\bigl(
    &\restr{(\polytag(X),\, \polycan(X),\, \polydcan(X),\, \polyproc(X) + \polysel(X))}{\Canb} \\
    &\quad \uplus\; \restr{(\polytag(X),\, \polynb(X),\, \polydnb(X),\, 0)}{\Nbb} , \\
    &\restr{(\polytag(X) - 1,\, \polycan(X),\, \polydcan(X),\, \polyproc(X))}{\Canb} \\
    &\quad \uplus\; \restr{(\polytag(X),\, \polyidrest(X),\, \polydrest(X),\, \polyprocrest(X))}{\Nbb}
  \bigr) .
\end{aligned}
\]
The edge consistency and distance verification constraints of \S\ref{sec:greedy-constraints} apply over the neighbor positions, while the candidate positions require no additional constraints of their own.
As the equality argument matches each identifier and its distance as one tuple, every candidate entry originates as a verified neighbor entry of an earlier timestep.

The sorting argument enforces that every region's candidate entries are in ascending order of distance.
We enforce nonnegativity of each difference $\polydcan(\omega X) - \polydcan(X)$ through a range check over $\Canbm$:
\[
  \Mem\bigl( \restr{(\polydcan(\omega X) - \polydcan(X))}{\Canbm},\; \Trange \bigr) .
\]

The truncation argument enforces that every discarded entry is at least as far from the query as the furthest retained one.
We introduce a helper polynomial $\polydworst(X)$, equal at every point of region $t$ to the last retained distance of region $t + 1$ through a copy constraint.
A membership argument into the range table enforces that every difference is nonnegative:
\[
  \Mem\bigl( \restr{(\polydrest(X) - \polydworst(X))}{\Nbb},\; \Trange \bigr) .
\]
Under the sorted order, $\polydworst(X)$ holds the furthest retained distance, and every discarded entry therefore lies at least as far from the query as every candidate of step $t + 1$.

Previously considered neighbors must not re-enter the candidate set, as the set could otherwise hold several copies of the same node.
We reuse the witness-flag mechanism of candidate selection, with a \textsf{considered} flag per neighbor supplied as the polynomial $\polycons(X)$, equal to $1$ exactly at the neighbors considered at an earlier step.
The identity
\[
\Bigl( \polydnbsub(X) - \bigl( 1 - \polycons(X) \bigr) \cdot \polydnb(X) - \polycons(X) \cdot \dinf \Bigr) \Big|_{\Nbb} = 0
\]
constrains the substituted distance $\polydnbsub(X)$ to the true distance at unconsidered neighbors and to $\dinf$ at considered ones, and $\polydnbsub(X)$ replaces $\polydnb(X)$ on the input side of the equality argument.
Considered neighbors thus exceed every retained distance and are routed into the discarded entries by the truncation, never re-entering the candidate set.
The considered flag is likewise an unconstrained witness, and we prove its consistency with the trace in the following.
\subsubsection{Flag Verification}
\label{sec:flag-verification}
A further source of complexity in proving the correctness of the beam search is that its choices at each timestep depend on which nodes were \textsf{considered} and which candidates \textsf{processed} at earlier timesteps.
Concretely, at each timestep the nearest unprocessed node in the candidate set is selected, and the unconsidered neighbors of the selected node are added to the set.
Following the standard treatment of data-dependent control flow in constraint systems~\cite{circ, cobbl}, the prover supplies two flags as part of the witness, a \textsf{processed} flag per candidate for the selection and a \textsf{considered} flag per neighbor for the insertion.
The flags must be proven consistent with the trace, since otherwise the prover can arbitrarily influence the search result, for example by falsely flagging a node as \textsf{considered} to exclude it from the result.
The difficulty in enforcing the correctness of the flags is that their true values depend on the execution's history, which naively requires verifying every timestep against all preceding ones, adding substantial prover overhead.
In this section, we show how to constrain the flags using a global property of the trace, rather than applying constraints at each individual timestep.

\fakeparagraph{Processed Flags}
For the \textsf{processed} flags, consistency with the trace reduces to consistency between adjacent timesteps.
A candidate's processed status changes only at its selection, and a correctly initialized flag consistent across adjacent timesteps is therefore consistent with the trace.
Concretely, an \emph{update} constraint enforces that the flag of the candidate selected for processing equals $1$, a \emph{preservation} constraint enforces that every other candidate's flag is unchanged between adjacent steps, and an \emph{initialization} constraint enforces that the flag of newly added candidates equals $0$.

We realize all three constraints by extending the equality argument of the candidate set update.
To identify the selected candidate within the argument, we introduce a selection polynomial $\polysel(X)$, whose evaluations are one-hot within each region, $1$ at the candidate selected for processing and $0$ at all others.
The argument relates tuples of tag, identifier, and distance, and we extend every tuple with the \textsf{processed} flag as a fourth component:
\[
\begin{aligned}
  \Eq\bigl(
    &\restr{(\polytag(X),\, \polycan(X),\, \polydcan(X),\, \polyproc(X) + \polysel(X))}{\Canb} \\
    &\quad \uplus\; \restr{(\polytag(X),\, \polynb(X),\, \polydnb(X),\, 0)}{\Nbb} , \\
    &\restr{(\polytag(X) - 1,\, \polycan(X),\, \polydcan(X),\, \polyproc(X))}{\Canb} \\
    &\quad \uplus\; \restr{(\polytag(X),\, \polyidrest(X),\, \polydrest(X),\, \polyprocrest(X))}{\Nbb}
  \bigr) .
\end{aligned}
\]
We set the candidates' flag component to $\polyproc(X) + \polysel(X)$ on the input side, updating the selected candidate's flag to $1$ while every other flag is unchanged, enforcing the \emph{update} and \emph{preservation} constraints.\footnote{The selected candidate's flag is necessarily $0$ before the update, as the candidate-selection constraints of \S\ref{sec:beam-constraints} enforce that the selected candidate is unprocessed.}
We fix the neighbor tuples' flag component to $0$, enforcing the \emph{initialization} constraint.

It remains to constrain $\polysel(X)$ to the correct one-hot encoding of the selected candidate.
To do so, it suffices to enforce (1) that $\polysel(X)$ equals $1$ at the selected candidate's point and (2) that $\polysel(X)$ equals $0$ at every other candidate point.
We enforce (1) by extending the selected-candidate membership argument of candidate selection with a third tuple component:
\[
  \Mem\bigl(
    \restr{(\polytag(X),\, \polyid(X),\, 1)}{\Selb},\;
    \restr{(\polytag(X),\, \polycan(X),\, \polysel(X))}{\Canb}
  \bigr) .
\]
Tuples match componentwise, and the constant $1$ on the query side forces the matched candidate entry, the selected candidate's, to equal $1$ in its $\polysel(X)$ component.
To enforce (2), we introduce the identity
\[
  \bigl( \polysel(X) \cdot ( \polyidsel(X) - \polycan(X) ) \bigr) \big|_{\Canb} = 0 ,
\]
where $\polyidsel(X)$ is copy-constrained to the selected candidate's identifier at every candidate point of the region.
At every candidate point, one of the two factors equals $0$, and $\polysel(X)$ equals $0$ at every candidate other than the selected one.

\fakeparagraph{Considered Flags}
While the \textsf{processed} state of a timestep is local to the timestep's candidate set, the \textsf{considered} state spans all nodes encountered at earlier timesteps.
Extending the transition-based verification of the \textsf{processed} flags to the \textsf{considered} flags would therefore require materializing the full set of encountered nodes at every timestep, increasing the evaluation domain's size from linear in the number of timesteps, $O(M \cdot T_b)$, to quadratic, $O(M \cdot T_b^2)$.
The set grows by a neighborhood of $2M$ per timestep, contributing $2M \cdot t$ evaluations at timestep $t$ and $2M \cdot T_b(T_b + 1)/2$ in total.

Instead of proving state transitions, we design constraints that enforce the flags' consistency directly against the trace and avoid resizing the evaluation domain.
It suffices to enforce, at every timestep $t$, that every node flagged $\textsf{considered} = 1$ and no node flagged $\textsf{considered} = 0$ appears in a neighborhood of an earlier timestep $t' < t$.
The $1$-direction is directly enforced through a membership argument against the neighborhoods of all timesteps.
A naive enforcement, however, instantiates one argument per timestep, each against the prefix of the trace preceding the timestep.
The $T_b$ arguments are each defined over the evaluation domain of size $O(M \cdot T_b)$, resulting in prover overhead quadratic in $T_b$.

We collapse the $T_b$ per-prefix arguments into a single argument against the full trace through timestep-tagging (\S\ref{sec:greedy-constraints}).
As every trace entry is tagged with its timestep, the prefix preceding timestep $t$ consists exactly of the entries tagged $t' < t$, and membership in the prefix is membership in the full trace under the tag condition.
The neighbor positions flagged \textsf{considered} form the witness set $\mathsf{Cons}_b$, with the selector $\polycons(X) \cdot \selq{\Nbb}(X)$.
The prover supplies as a witness the timestep $t'(X)$ at which each node was considered, and a single membership argument over $\mathsf{Cons}_b$ against the full trace, with a range check enforcing $t' < t$, enforces the $1$-direction:
\begin{align*}
  \Mem\bigl( \restr{(\polynb(X),\, t'(X))}{\mathsf{Cons}_b},\; \restr{(\polynb(X),\, \polytag(X))}{\Nbb} \bigr) ,\\
  \Mem\bigl( \restr{(\polytag(X) - t'(X) - 1)}{\mathsf{Cons}_b},\; \Trange \bigr) .
\end{align*}

While the $1$-direction is a set membership claim, the $0$-direction is a non-membership claim, which cannot be efficiently established through lookups alone.
Set membership is an existential claim that a membership argument enforces through a witness attesting the queried value's occurrences in the table (e.g.\ the multiplicities $m_j$ in LogUp).
Enforcing the $0$-direction would instead require negating the existential (NP) claim, turning it into a universal (coNP) one that must hold for all witnesses.
Existing zkSNARKs are defined for existential (NP) relations, and non-membership is therefore typically reduced to an existential claim through one of two approaches, each scaling the evaluation domain's size quadratically in the number of timesteps.
The first approach replaces non-membership in a set with membership in the set's complement with respect to a fixed universe of elements.
In our setting, the universe consists of all $2M \cdot T_b$ nodes encountered in the trace, and materializing it at each of the $T_b$ timesteps requires $2M \cdot T_b^2$ evaluations.
The second approach is to keep the set sorted and establish non-membership of a value by supplying the gap between consecutive entries in which the value lies.
In our setting, the considered set grows by $2M$ nodes per timestep, and restating it in sorted order at each timestep results in $2M \cdot T_b(T_b + 1)/2$ evaluations in total.

We instead identify a global property of the trace that suffices for the $0$-direction and can be enforced over the existing evaluation domain.
In an execution of the algorithm, a node is not considered until its first appearance on the trace.
The node identifiers flagged $\textsf{considered} = 0$ are therefore distinct on the trace.
We name the property the \emph{uniqueness} of unconsidered nodes and enforce it directly, in place of the per-timestep non-membership checks.
Enforcing the \emph{uniqueness} of unconsidered nodes reduces to a standard sorting argument, at a constant number of additional polynomials and constraints.

The sorting argument encodes the $\textsf{considered} = 0$ node identifiers in sorted order, encoding the flag--identifier pairs $(\polycons(X), \polynb(X))$ as two witness polynomials $(s(X), u(X))$, sorted by flag first and identifier within each flag value.
At positions outside $\Nbb$, the prover sets the flag to $1$, and the corresponding pairs sort into the tail of the column.
The prover supplies $s(X)$ and $u(X)$ as witnesses, an equality argument enforces that their evaluations are a reordering of the pairs, a polynomial identity enforces that $s(X)$ is nondecreasing, and a range check enforces strict monotonicity between adjacent entries of $u(X)$ over the $s(X) = 0$ prefix:
\begin{align*}
  &\Eq\bigl( (\polycons(X),\, \polynb(X)),\; (s(X),\, u(X)) \bigr) , \\
  &\bigl( (s(\omega X) - s(X)) \cdot (s(\omega X) - s(X) - 1) \bigr) \big|_{H^{-}} = 0 , \\
  &\Mem\bigl( \restr{(u(\omega X) - u(X) - 1)}{(1 - s(X))(1 - s(\omega X)) \cdot \selq{H^{-}}},\; \Trange \bigr) .
\end{align*}
As the $\textsf{considered} = 0$ identifiers occupy the prefix of $u(X)$, distinctness of adjacent entries implies distinctness of all entries, and strict monotonicity enforces the \emph{uniqueness}.

%% file: sections/evaluation.tex
\section{Evaluation}
\label{sec:evaluation}
We evaluate whether \oursystem makes verifiable HNSW search practical without sacrificing retrieval quality. Concretely, this section answers three questions: (i) does \oursystem preserve the recall of plaintext HNSW (\S\ref{sec:eval-approximation})? (ii) what is the cost of verifiability, and how do proving time, proof size, and verifier time scale with corpus size and embedding dimension (\S\ref{sec:eval-performance})? and (iii) how does \oursystem compare against other verifiable semantic search systems (\S\ref{sec:eval-performance})?

\subsection{Experimental Setup}
\label{sec:eval-setup}
\paragraph{Datasets.}

We evaluate on SIFT1M together with three larger prefixes of SIFT1B, the
BIGANN datasets of 10M, 50M, and 100M vectors. These share SIFT1M's
128-dimensional integer descriptors and let us measure cost as the index
grows from one million to one hundred million vectors.
We also evaluate on Deep10M and GIST1M, two floating-point datasets that
exercise different regimes. Deep10M consists of 96-dimensional neural
network features, representative of the learned embeddings that semantic
search and RAG systems retrieve over. GIST1M consists of 960-dimensional
GIST image descriptors, representative of high-dimensional retrieval
workloads. Together, the integer and floating-point datasets let us
separate the cost of the fixed-step approximation from that of quantization.

\paragraph{Implementation.}
We implement \oursystem in Rust on top of the halo2-axiom library~\cite{halo2-axiom}, over the BN254 curve. We extend the library with an
implementation of the cq lookup argument~\cite{cq} and accelerate the cq
preprocessing on the GPU using the icicle library~\cite{icicle}. We build
each HNSW index with FAISS~v1.14.2 at
\texttt{efConstruction}${}=800$, except for BIGANN-50M
and BIGANN-100M, where we use \texttt{efConstruction}${=}200$.  We use FAISS as the plaintext reference for retrieval quality. Embeddings are quantized to 8-bit field elements, and
all distances are squared $\ell_2$.
We run all experiments on a machine with an AMD EPYC 9654 96-core
processor, 768\,GB of RAM, and an NVIDIA H100 GPU.
Proving and verification run on the CPU; preprocessing is accelerated on the GPU.
The comparison of systems in Figure~\ref{fig:recall-prover} runs on hardware chosen to match the published environment of each system as closely as possible, as stated inline.

\paragraph{Metrics.}
For each benchmark, we evaluate on its standard query set and measure retrieval quality as recall@1, the fraction of queries for which the search returns the true nearest neighbor.
The ground truth is the nearest neighbor in the original, unquantized embedding space.
The reported recall therefore reflects the loss from both quantization and truncation. 
We report proving time, proof size, and verification time per query, along with the distributions of greedy and beam steps each search takes, which indicate the recall cost of any given choice of $T_g$ and $T_b$.

\subsection{Accuracy}
\label{sec:eval-approximation}
\oursystem proves a fixed-budget search over quantized embeddings, and we evaluate the impact of the two approximations on retrieval quality.
For each configuration, we compute the queries using the FAISS library and record the number of greedy and beam steps each query takes.
We then evaluate the quality of \oursystem' search with the step budget fixed at different quantiles of this distribution.
Figure~\ref{fig:truncation-recall} reports recall@1 at each budget against the unbounded float32 FAISS baseline.

\paragraph{Cost of Truncation.}
Recall is preserved when truncating the search well below the maximum observed steps, implying the budget does not need to cover the longest-running queries.
Fixing the budget at the 95th percentile shortens the proof trace by 13\% to 57\% relative to the maximum.
Recall@1 stays within 0.8 points of FAISS at every integer-dataset configuration, and several configurations match the baseline within query-set noise.
The budget itself is tied to the search parameters rather than the corpus, as the beam steps depend primarily on $\mathit{ef}$ and the corpus size contributes only logarithmically through the layer count.
At $(M, \mathit{ef}) = (16, 32)$, the 95th-percentile beam step count rises from 38 on SIFT1M to only 44 on the hundredfold larger BIGANN-100M.

\paragraph{Cost of Quantization.}
On the floating-point datasets Deep10M and GIST1M, whose real-valued embeddings are converted to 8-bit integers, quantization introduces an additional recall loss.
We isolate this loss by measuring at the maximum budget, where the fixed-step search is equivalent to unbounded HNSW.
The remaining gap stems from quantization alone and reaches 5.2 recall@1 points on GIST1M at $(M, \mathit{ef}) = (32, 64)$ (87.4\% against 92.6\%).

\begin{figure}[t]
\centering
\includegraphics[width=\columnwidth]{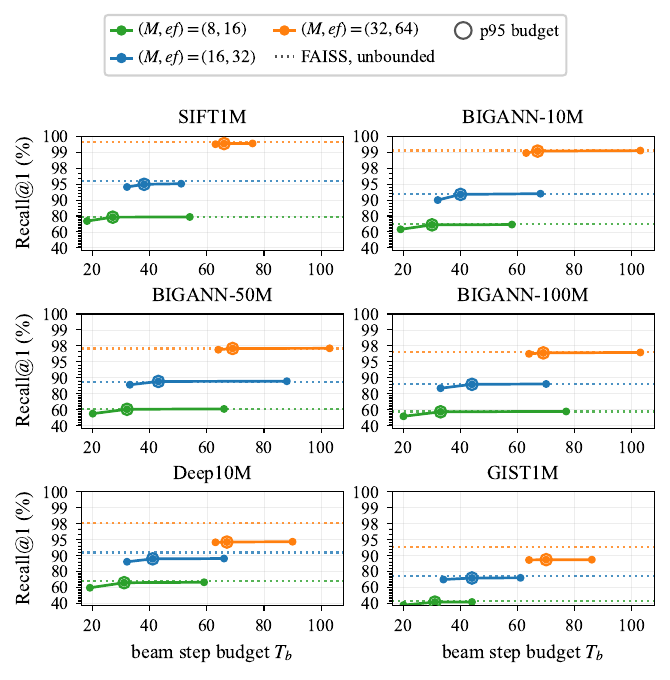}
\caption{Recall@1 of the proven fixed-budget search across beam step budgets $T_b$, measured at the p50, p95, and maximum budget of each configuration.}
\label{fig:truncation-recall}
\end{figure}

\subsection{Performance}
\label{sec:eval-performance}
We evaluate the cost of verifiability for recall@1 above 0.9, as this is a standard target for practical retrieval.
For each dataset, we fix the configuration to the smallest setting that achieves the target recall and report proving time, proof size, and verification time per query in Table~\ref{tab:performance}.
\oursystem' constraint system is determined by the parameters $d$, $\mathit{ef}$, $M$, $T_g$, and $T_b$, and proving cost therefore depends on the configuration required to reach the target recall.

\paragraph{Scaling with Dataset Size.}
At a fixed recall target, proving time increases only modestly as the dataset grows from one to one hundred million vectors.
\oursystem proves a query on SIFT1M in 0.80 seconds, and BIGANN takes 1.22, 1.99, and 1.98 seconds at 10M, 50M, and 100M vectors, respectively.
The corpus size affects proving cost only through the search configuration, as a larger corpus requires a larger $\mathit{ef}$ to reach the recall target, which rises from 26 to 48 while $M = 16$ remains fixed (Table~\ref{tab:performance}).
The proving times at 50M and 100M vectors coincide, as the traces of both configurations fit within the same evaluation domain.

\paragraph{Scaling with Embedding Dimension.}
Beyond the search configuration, proving cost depends on the embedding dimension $d$, as every distance computation in the trace sums over $d$ coordinates, each represented by a separate polynomial.
For GIST1M ($d = 960$), \oursystem proves a query in 36.66 seconds with a 75.5~kB proof, a cost that combines the high dimension with the largest configuration required to reach the recall target (Table~\ref{tab:performance}).
Verification time depends on $d$ as well, taking less than 71 milliseconds on every dataset with $d \leq 128$ but reaching 1.94 seconds on GIST1M.
Dimensionality-reduction techniques make the projected dimension another tunable parameter that trades answer quality for proving cost, which we quantify in our RAG evaluation below. 

\begin{table}[t]
\centering
\footnotesize
\setlength{\tabcolsep}{3pt}
\begin{tabular}{l r r r r r r r}
\toprule
Dataset & $M$ & $\mathit{ef}$ & $T_g$ & $T_b$ & Prove (s) & Verify (ms) & Proof (kB) \\
\midrule
SIFT1M      & 16 &  26 &  6 &  26 &  0.80 &   40 & 16.5 \\
BIGANN-10M  & 16 &  32 & 14 &  33 &  1.22 &   50 & 16.5 \\
Deep10M     & 16 &  40 & 21 &  54 &  1.83 &   58 & 14.5 \\
BIGANN-50M  & 16 &  40 & 20 &  50 &  1.99 &   71 & 16.5 \\
BIGANN-100M & 16 &  48 & 21 &  59 &  1.98 &   70 & 16.5 \\
GIST1M      & 64 & 128 &  9 & 127 & 36.66 & 1940 & 75.5 \\
\bottomrule
\end{tabular}
\caption{Prover performance of \oursystem at the smallest configuration reaching recall@1 above 0.9 on each dataset, with the step budgets $(T_g, T_b)$: proving time, verification time, and proof size per query.}
\label{tab:performance}
\end{table}

\begin{figure}[!t]
\centering
\includegraphics[width=\columnwidth]{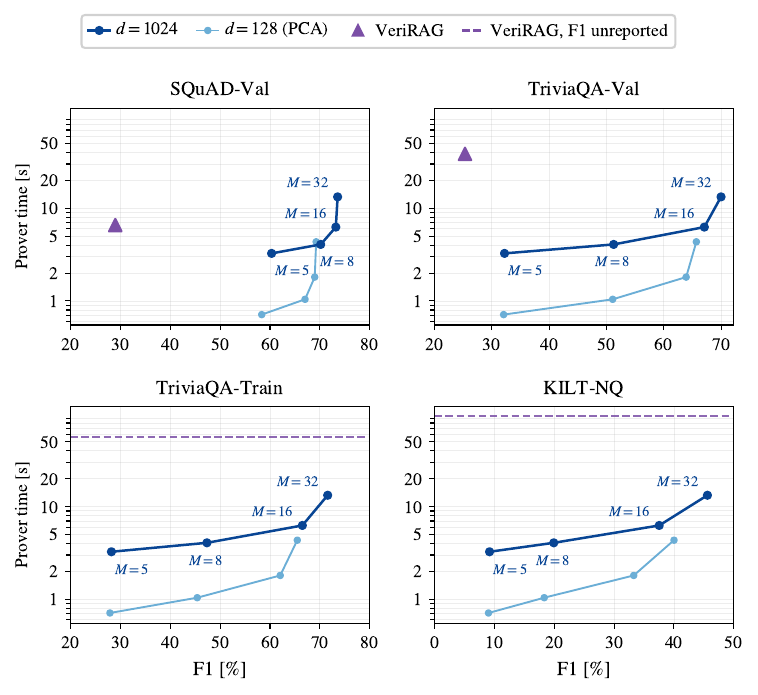}
\caption{F1 versus per-query proving time on SQuAD, TriviaQA-Val, TriviaQA-Train, and KILT-NQ at 256-token chunks, log scale.}
\label{fig:rag-frontier}
\end{figure}
\paragraph{Comparison to Verifiable Search Systems.}

\oursystem' proving cost is lower than that of prior verifiable search systems at every recall level (Figure~\ref{fig:recall-prover}).
We first compare against zkRAG~\cite{zkrag}, concurrent work that likewise proves HNSW search and reports only single-threaded prover times.
For the same parameters, $(M, \mathit{ef}) = (32, 64)$, and the same number of processed level-0 nodes ($N_{\mathrm{exp}} = 2^{12}$), \oursystem proves a query 2.8$\times$ faster, in 18.6 seconds on one thread of a Xeon 8151 against zkRAG's 51.5 seconds on one thread of a Xeon 6126, two processors within 10\% in single-core performance.
The speedup comes alongside a stronger guarantee, as zkRAG's proof reveals the number of steps taken at every layer while \oursystem' reveals nothing beyond the result.
\oursystem' prover further benefits from multithreaded acceleration, with 16 threads reducing proving time 6.6$\times$, from 18.6 to 2.8 seconds.
In the multithreaded setting, \oursystem also outperforms V3DB~\cite{v3db}, an IVF-PQ system, reaching every recall level V3DB attains between 10$\times$ and 45$\times$ faster, as product quantization caps V3DB's recall@1 at 0.50, a ceiling \oursystem exceeds in 0.64 seconds against V3DB's 29.2 seconds.

\paragraph{Application: End-to-end Verifiable RAG.}
We build a complete RAG pipeline with \oursystem and measure how answer quality and proving cost vary with the search configuration.
We instantiate \oursystem' proven fixed-budget HNSW search as the retrieval component, and the remaining components follow VeriRAG's evaluation setting, with BGE-M3~\cite{bgem3} for embedding and Qwen2-7B~\cite{qwen2} for generation.
We evaluate on SQuAD~\cite{squad}, TriviaQA~\cite{triviaqa}, and KILT~\cite{kilt} at 256-token chunks, reporting F1 over the full validation sets for $M = 5$ to $M = 32$ with $\mathit{ef} = 2M$ (Figure~\ref{fig:rag-frontier}).
Unlike the percentile budgets of \S\ref{sec:eval-approximation}, these experiments fix a single corpus-independent budget $T_g = T_b = \mathit{ef}$.
At $M = 32$, \oursystem reaches 73.6 F1 on SQuAD, 70.0 on TriviaQA-Val, 71.6 on TriviaQA-Train, and 45.6 on KILT-NQ, at 13.3 seconds of prover time per query.
Projecting the embeddings to $d = 128$ with principal component analysis reduces proving time 3$\times$, to 4.3 seconds, at a cost of 3.5 to 6.1 F1 points across the datasets.
We compare against VeriRAG on a 24-core Xeon Platinum 8175M, comparable to the 24-core Xeon Gold 5220 used in their evaluation.
\oursystem exceeds VeriRAG's F1 at a fraction of its proving time, reaching 60.3 F1 in 2.5 seconds on SQuAD against VeriRAG's 46.9 in 6.6, and 51.3 in 3.2 seconds on TriviaQA-Val against 41.1 in 38.7.
The gap follows from the retrieval quality bottleneck of IVF-PQ, as the cluster-based index recovers fewer of the most relevant passages and degrades the context the generator receives (\S\ref{sec:related}).
\oursystem' proving cost further depends on the search configuration rather than the corpus size, while VeriRAG's grows directly with the corpus, from 6.6 seconds on SQuAD to 56.0 on TriviaQA-Train and 96.2 on KILT.

\begin{figure}[!t]
\centering
\includegraphics[width=\columnwidth]{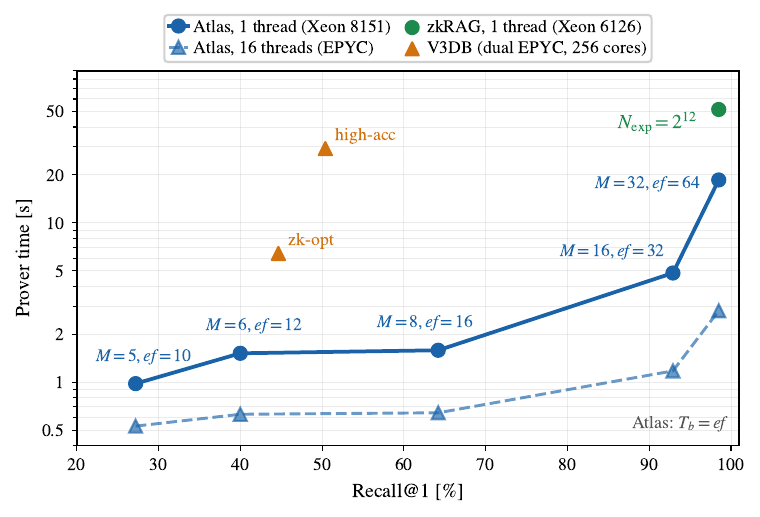}
\caption{Recall@1 versus per-query proving time on SIFT1M, log scale.}
\label{fig:recall-prover}
\end{figure}

%% file: sections/appendix.tex
\section{Extended Background Material}
\label{app:cq}

\fakeparagraph{Log-Derivative Reduction}
Verifying the product relation of \S\ref{sec:bg-arguments} is expensive, particularly when the table is large.
Haböck~\cite{haboeck2022logup} observed that taking the logarithmic derivative of both sides (i.e.\ applying $\frac{d}{dX}\log(\cdot)$) converts the product of linear factors into a sum of their reciprocals.
The product equality holds if and only if the resulting sums are equal.
By the Schwartz--Zippel lemma, equality of the two sums reduces to their equality at a random challenge $\gamma \xleftarrow{\$} \mathbb{F}$:
\begin{equation}
\label{eq:logderiv}
  \sum_{i=0}^{n-1} \frac{1}{\gamma + f(\omega_H^i)}
  = \sum_{j=0}^{N-1} \frac{m(\omega_K^j)}{\gamma + t(\omega_K^j)}
\end{equation}
where $m(X)$ is the polynomial over $K$ encoding the multiplicities.
This reduction from a product to a sum is the basis of the cq lookup argument~\cite{cq}.

\fakeparagraph{The cq Protocol}
cq~\cite{cq} is a lookup argument designed for tables far larger than the query.
It precomputes the table-side sum of the log-derivative identity, and its per-query proving cost therefore scales with the size of the query side alone.
Concretely, cq realizes the log-derivative check~\eqref{eq:logderiv} as a polynomial protocol over KZG commitments~\cite{kzg}.

\fakeparagraph{Preprocessing}
During setup, KZG commitments to the Lagrange basis polynomials $\{L_j^K(X)\}_{j=0}^{N-1}$ over $K$ are computed.
For each table entry $t(\omega_K^j)$, a cached quotient $q_j(X)$ is derived from $L_j^K(X)$ and $t(X)$ such that $Q_A$ (defined below) can later be expressed as a linear combination of the $q_j$.
KZG commitments to the $q_j$ are computed and stored.
This costs $O(N \log N)$ and depends only on the table.

\fakeparagraph{Proving}
The prover constructs two rational functions encoding each side of~\eqref{eq:logderiv}:
\[
  B(X) = \sum_{i=0}^{n-1} \frac{L_i^H(X)}{\gamma + f(\omega_H^i)},
  \qquad
  A(X) = \sum_{j=0}^{N-1} \frac{m_j \cdot L_j^K(X)}{\gamma + t(\omega_K^j)} .
\]
The prover commits to $m(X) = \sum_j m_j L_j^K(X)$, the multiplicity polynomial.
Since only $n$ of the $N$ multiplicities are nonzero, the commitment is computed as a sparse linear combination of the preprocessed Lagrange basis commitments in $O(n)$ scalar multiplications.
The prover also commits to $B(X)$, computed over $H$ in $O(n \log n)$.

The protocol verifies three polynomial identities:
\begin{enumerate}
\item \emph{Well-formedness of $B$.} $B(X)(\gamma + f(X)) - 1 = Q_B(X) \cdot Z_H(X)$, ensuring $B(\omega_H^i) = \frac{1}{\gamma + f(\omega_H^i)}$ for all $i$.
The quotient $Q_B$ costs $O(n \log n)$.
\item \emph{Well-formedness of $A$.} $A(X)(\gamma + t(X)) - m(X) = Q_A(X) \cdot Z_K(X)$, ensuring $A(\omega_K^j) = \frac{m_j}{\gamma + t(\omega_K^j)}$ for all $j$.
The commitment to $Q_A$ is computed as a sparse linear combination of the preprocessed cached quotient commitments in $O(n)$~\cite{cq}.
\item \emph{Sum equality.} $\sum_{x \in H} B(x) = \sum_{x \in K} A(x)$, verified via a univariate sumcheck.
The sumcheck quotient $Q_C$ costs $O(n \log n)$.
\end{enumerate}

\allowdisplaybreaks
\section{Equivalence Proofs}
\subsection{Equivalence of the Unified Greedy Search}
\label{app:proof-greedy}

\begin{theorem}[Equivalence of $G_{\mathrm{upper}}$]
\label{thm:greedy-equivalence}
Let $G$ be a multi-layer graph with layers $G_0, G_1, \ldots, G_L$, each a graph $G_\ell$ with vertex set $V(G_\ell)$ and edge set $E(G_\ell)$, let $q \in \mathbb{R}^d$ be a query, and let $\mathit{ep}$ be an entry point in $G_L$.
Let $v^*_1$ be the result of the greedy search of the HNSW algorithm as
in lines~\ref{ln:epset}--\ref{ln:eplast} of~\Cref{alg:hnsw_search}, and
let $T^*_\ell$ denote the number of greedy steps taken at layer $\ell$.
Define $G_\mathrm{upper} = (V_\mathrm{upper}, E_\mathrm{upper})$ by
\begin{align*}
V_\mathrm{upper} ={}& \bigl\{(v, \ell) : \ell \in [L],\,
    v \in V(G_\ell)\bigr\} \cup \bigl\{(v, 0) : v \in V(G_1)\bigr\},\\
E_\mathrm{upper} ={}& \bigl\{((u, \ell), (w, \ell)) : \ell \in [L],\,
    (u, w) \in E(G_\ell)\bigr\}\\
{}\cup{}& \bigl\{((v, \ell), (v, \ell{-}1)) : \ell \in [L],\,
    v \in V(G_\ell)\bigr\},
\end{align*}
i.e., each layer contributes a labeled copy of its edges, and
drop-down edges connect every node to its copy one layer below.
Then the single invocation of greedy search
$\textsc{Search-Layer}'_{1}(G_\mathrm{upper}, q, \{(\mathit{ep}, L)\})$
terminates at $(v^*_1, 0)$.
The total number of traversal steps taken by $\textsc{Search-Layer}'$ is
$T_g = \sum_{\ell=1}^{L} T^*_\ell + L$.
\end{theorem}

\begin{proof}
By downward induction on $\ell$.
At layer $\ell \geq 1$, the neighbors of the current node $(v, \ell)$
in $G_{\mathrm{upper}}$ are its intra-layer neighbors and the
drop-down $(v, \ell{-}1)$.
Since $(v, \ell)$ and $(v, \ell{-}1)$ share the same embedding,
$d((v, \ell{-}1), q) = d((v, \ell), q)$, so under the lexicographic
comparison of $\textsc{Search-Layer}'$ the drop-down cannot be the
greedy choice while any strictly closer intra-layer neighbor exists.
$\textsc{Search-Layer}'$ therefore reproduces
$\textsc{Search-Layer}_1(G_\ell, q, \cdot)$ step by step until
reaching the local minimum $v^*_\ell$, taking exactly $T^*_\ell$ steps.
At that point all intra-layer neighbors have distance
$\geq d(v^*_\ell, q)$, and since ties are broken in favor of the
lower layer, the drop-down is the greedy choice and the walk descends
to $(v^*_\ell, \ell{-}1)$.
Repeating at each layer down to $\ell = 1$, the walk reaches
$(v^*_1, 0)$, which has no outgoing edges, so the search terminates
there.
In total, the walk takes $T^*_\ell$ intra-layer steps at each layer
plus one drop-down per layer, giving
$T_g = \sum_{\ell=1}^{L} T^*_\ell + L$ steps.
\end{proof}

\subsection{Proof of Lemma~\ref{lem:order-independence}}
\label{app:proof-order}
\begin{proof}
Let $T = \mathrm{top}_{\mathit{ef}}(W \cup A)$ and let $W_{\mathrm{final}}$ denote $W$ after all insertions.
If $|W \cup A| \le \mathit{ef}$, no eviction occurs and $W_{\mathrm{final}} = W \cup A = T$.
Otherwise, eviction occurs when an insertion brings $|W|$ to $\mathit{ef} + 1$ and removes the furthest element.
Any $e \in T$ has at most $\mathit{ef} - 1$ elements closer to $q$ in all of $W \cup A$, so among the $\mathit{ef} + 1$ elements in $W$ at the moment of eviction, at least one is further than $e$, and $e$ is never the one removed.
Therefore $T \subseteq W_{\mathrm{final}}$, and as $|W_{\mathrm{final}}| = \mathit{ef} = |T|$, $W_{\mathrm{final}} = T$.
\end{proof}

\subsection{Proof of Lemma~\ref{lem:pop-invariance}}
\label{app:proof-invariance}
\begin{proof}
At $t = 0$, $W_0 = C_0 = \{\mathit{ep}\}$ in both executions.
At each subsequent step, $W_t$ agrees by
Lemma~\ref{lem:order-independence}, so the same neighbors are processed and the same nodes enter both $C$ and $W$.
Let $C_t, C'_t$ denote the candidate sets in the two
executions.
$C_t \triangle C'_t$ consists of nodes that were admitted into
$W$ temporarily and then evicted, so for any
$e \in C_t \triangle C'_t$,
$d(e, q) > \max_{w \in W_t} d(w, q)$, and $e$ triggers
termination at line~\ref{ln:break} before being extracted,
leaving the same node to be extracted at line~\ref{ln:pop} in
both executions.
\end{proof}

\subsection{Proof of \Cref{thm:beam-equivalence}}
\label{app:proof-beam}
\begin{proof}[Proof sketch]
We show that the selected node and the state of $W$ agree at every iteration across both algorithms.
By Lemma~\ref{lem:pop-invariance}, the sequence of selected nodes in \textsc{Search-Layer} is invariant under insertion order.
Phase~2 reproduces this sequence: the selection rule $\arg\min\{d_i : p_i = 0\}$ picks the same node that \textsc{Search-Layer} would extract from $C$ (line~\ref{ln:pop}), and the batched merge produces the same $W$ as sequential insertion by Lemma~\ref{lem:order-independence}.
Since the selected node and $W$ agree at every iteration, the final output is identical.
Once every candidate in $W$ is processed, \textsc{Search-Layer}'s next extracted node is further than $W$'s furthest element and the search terminates, while Phase~2 finds no unprocessed candidate and leaves $W$ unchanged, so any budget $T_b \geq T^*$ yields the same final set.
\end{proof}
\begin{figure}[!b]
\centering
\includegraphics[width=\columnwidth]{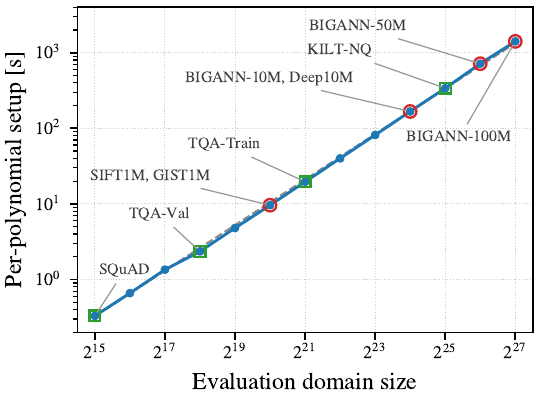}
\caption{One-time cq preprocessing per committed polynomial across evaluation domain sizes $N$, measured on the H100, with the evaluated datasets marked at their domains. The dashed line is $10^{-5} \cdot N$ seconds.}
\label{fig:preprocessing}
\end{figure}
\section{Preprocessing}
Committing the embedding and edge tables is a one-time cost per index.
\Cref{fig:preprocessing} reports the preprocessing time per committed polynomial, the generation of one column's cached quotients, across domain sizes from $2^{15}$ to $2^{27}$, accelerated on the H100 GPU through icicle~\cite{icicle}.
The cost grows near-linearly with the domain size, at roughly ten microseconds per table entry.
For the evaluated datasets, one polynomial costs 9.6 seconds on the million-node SIFT1M and GIST1M, 167 seconds for BIGANN-10M and Deep10M, and 716 and 1411 seconds for BIGANN-50M and BIGANN-100M, whose tables fit within the $2^{20}$, $2^{24}$, $2^{26}$, and $2^{27}$ domains.
The RAG corpora of \S\ref{sec:eval-performance}, chunked at 256 tokens, span ${\sim}20$k chunks for SQuAD, ${\sim}250$k for TriviaQA-Val, ${\sim}2$M for TriviaQA-Train, and ${\sim}30$M for KILT-NQ, which fit within the $2^{15}$, $2^{18}$, $2^{21}$, and $2^{25}$ domains.
At these sizes, a committed polynomial costs 0.33 seconds for SQuAD, 2.4 seconds for TriviaQA-Val, 19.6 seconds for TriviaQA-Train, and 339 seconds for KILT-NQ.